\documentclass{llncs}
\usepackage{tabularx,booktabs,multirow,delarray,array}
\usepackage{graphicx,amsmath,amssymb}
\usepackage{latexsym}
\def\calP{\mathcal{P}}
\def\calY{\mathcal{Y}}
\def\calX{\mathcal{X}}

\def\calL{\mathcal{L}}

\def\calH{\mathcal{H}}

\def\bbR{\mathbb{R}}
\def\Ed{\mathsf{E}\mathrm{d}}

\renewcommand{\figurename}{Figure}
\newtheorem{observation}{Observation}

\begin{document}

\title{Computing the Rectilinear Center of Uncertain Points in the Plane}

\author{Haitao Wang
\and
Jingru Zhang
}

\institute{
Department of Computer Science\\
Utah State University, Logan, UT 84322, USA\\
\email{haitao.wang@usu.edu,jingruzhang@aggiemail.usu.edu}\\
}

\maketitle

\pagestyle{plain}
\pagenumbering{arabic}
\setcounter{page}{1}

%\vspace{-0.15in}
\begin{abstract}
In this paper, we consider the rectilinear one-center problem on
uncertain points in the plane. In this problem,
we are given a set $\calP$ of $n$ (weighted) uncertain points in the plane
and each uncertain point has $m$ possible locations each associated
with a probability for the point appearing at that location.
The goal is to find a point $q^*$ in the plane which minimizes the
maximum expected rectilinear distance from $q^*$ to all uncertain points
of $\calP$, and $q^*$ is called a {\em rectilinear center}.
We present an algorithm that solves the problem
in $O(mn)$ time.  Since the input size of the problem
is $\Theta(mn)$, our algorithm is optimal.
\end{abstract}

%\vspace{-0.2in}
\section{Introduction}\label{Sec:1}
In the real world, data is inherently uncertain due to many facts, such as
the measurement inaccuracy, sampling discrepancy, resource limitation, and so on.
A large amount of work has recently been done on uncertain data,
e.g., \cite{ref:AgarwalIn09,ref:AgarwalNe12,ref:AgarwalCo14,ref:DongDa07,ref:KamousiCl11,ref:KamousiSt11,ref:SuriOn14,ref:SuriOn13}.
%The uncertain data thus has already been considered in a lot of classical problems,
%like data clustering \cite{ref:AggarwalFr08,ref:ChauUn06,ref:CormodeAp08},
%data indexing \cite{ref:AgarwalIn09,ref:ChengEf04,ref:TaoIn05}
%and nearest-neighbor searching \cite{ref:AgarwalNe13,ref:AgarwalNe12,ref:ChengEv09}.
In this paper, we study the one-center problem on uncertain points in the plane with
respect to the rectilinear distance.
%We introduce the problem below.

Let $\calP=\{P_1,P_2,\ldots, P_n\}$ be a set of $n$ uncertain points in the plane,
where each uncertain point $P_i\in \calP$ has $m$ possible locations
$p_{i1}, p_{i2}, \cdots, p_{im}$ and for each $1\leq j\leq m$, $p_{ij}$ is
associated with a probability $f_{ij}\geq 0$ for $P_i$ being at $p_{ij}$ (which
is independent of other locations).  %with $\sum_{i=1}^mf_{ij}=1$.

For any (deterministic) point $p$ in the plane, we use $x_p$ and $y_p$ to denote
the $x$- and $y$-coordinates of $p$, respectively. For any two points $p$ and
$q$, we use $d(p,q)$ to denote the {\em rectilinear distance} between $p$ and
$q$, i.e., $d(p,q)=|x_p-x_q|+|y_p-y_q|$.

Consider a point $q$ in the plane. For any uncertain point $P_i\in \calP$, the
{\em expected rectilinear distance} between $q$ and $P_i$ is defined as
$$\Ed(P_i, q)=\sum_{j=1}^{m}f_{ij}\cdot d(p_{ij},q).$$
%In fact, each $\Ed(P_i,q)$ varies only with the values of $x_q$ and $y_q$,
%and hence, to be explicit, we use $\Ed_i(x_q,y_q)$ to represent $\Ed(P_i,q)$.

Let $\Ed_{\max}(q)=\max_{P_i\in \calP}\Ed(P_i,q)$. A point $q^*$ is called a {\em
rectilinear center} of $\calP$ if it minimizes the value $\Ed_{\max}(q^*)$ among all
points in the plane.  Our goal is to compute $q^*$. Note that such a point $q^*$ may not be unique, in which case we let $q^*$ denote an arbitrary such point.

We assume that for each uncertain point $P_i$ of $\calP$,
its $m$ locations are given in two sorted lists,
one by $x$-coordinates and the other by $y$-coordinates.
To the best of our knowledge, this problem has not been studied before. In this
paper, we present an $O(mn)$ time algorithm. Since the input
size of the problem is $\Theta(nm)$, our algorithm essentially
runs in linear time, which is optimal.

Further, our algorithm is applicable to the weighted version of this problem
in which each $P_i\in \calP$ has a weight $w_i\geq 0$ and the
{\em weighted expected distance}, i.e., $w_i\cdot \Ed(P_i,q)$, is
considered. To solve the weighted version, we can first
reduce it to the unweighted version by changing each $f_{ij}$ to
$w_i\cdot f_{ij}$ for all $1\leq i\leq n$ and $1\leq j\leq m$, and
then apply our algorithm for the unweighted version. The running
time is still $O(mn)$.
%to solve the weighted version in $O(mn)$ time as well.
%Note that although we have $\sum_{i=1}^mf_{ij}=1$ in the input for each $P_i\in \calP$ (which means $P_i$ will appear at some location), our algorithm still works if the equality does not hold.

\subsection{Related Work}
\label{sec:related}

%Two models on uncertain data have been commonly considered:  the \textit{existential} model
%\cite{ref:AgarwalCo14,ref:KamousiCl11,ref:KamousiSt11,ref:SuriOn14,ref:SuriOn13} and the \textit{locational} model \cite{ref:AgarwalIn09,ref:AgarwalNe12,ref:ChengEf04,ref:TaoRa07}. In the existential
%model, an uncertain point has a specific location but its existence is uncertain.
%In the locational model, an uncertain point always exists but its location is uncertain and follows a probability distribution function. Our one-center problem belongs to the locational model. In fact, the same problem under existential model is essentially the weighted one-center problem
%for deterministic data, which was solved in linear time \cite{ref:MegiddoLi83}.

The problem of finding one-center among uncertain points on a line has
been considered in our previous work
\cite{ref:WangOn15}, where an $O(nm)$ time algorithm was given.
An algorithm for computing $k$ centers for
general $k$ was also given in \cite{ref:WangOn15} with the running time
$O(mn\log mn+n\log n\log k)$. In fact, in \cite{ref:WangOn15}
we considered the $k$-center problem under a
more general uncertain model where each uncertain point can appear in
$m$ intervals.  We also studied the one-center problem for uncertain points on tree
networks in \cite{ref:WangCo15}, where a linear-time algorithm was
proposed.

%In \cite{ref:WangCo15}, each uncertain point may appear in $m$ deterministic
%locations on the given tree with probabilities and the center of the $n$ uncertain
%points on the given tree can be found
%in the $O(mn)$ time (we consider the general case where the center and the locations
%does not necessarily lay in the vertices of the tree).

There is also a lot of other work on facility location problems for uncertain data.
For instances, Cormode and McGregor \cite{ref:CormodeAp08} proved that the $k$-center problem
on uncertain points each associated with multiple locations in high-dimension space
is NP-hard and gave approximation algorithms for different problem models.
Foul \cite{ref:FoulA106} considered the Euclidean one-center problem on uncertain points
each of which has a uniform distribution in a given rectangle in the plane.
de Berg. et al. \cite{ref:BergKi13} studied the Euclidean 2-center problem
for a set of moving points in the plane (the moving points can be considered uncertain).
%and proposed an approximation algorithm.
%See other facility location problems under uncertainty,
%e.g., the minimax regret problems \cite{ref:AverbakhMi00,ref:AverbakhFa05,ref:WangMi14}.

The $k$-center problems on deterministic points
are classical problems and have been studied extensively.
When all points are in the plane, the problems on most distance metrics
are NP-hard \cite{ref:MegiddoOn84}. However,
some special cases can be solved in polynomial time,
e.g., the one-center problem \cite{ref:MegiddoLi83},
the two-center problem \cite{ref:ChanMo99},
the rectilinear three-center problem \cite{ref:HoffmannAs05},
the line-constrained $k$-center problems
(where all centers are restricted to be on a given line in the plane) \cite{ref:BrassTh11,ref:KarmakarSo13,ref:WangLi14}.
%In addition, polynomial algorithms \cite{ref:MegiddoAn81,ref:MegiddoNe83}
%were proposed for the $k$-center problems on tree networks and
%the one-dimensional version is solved in $O(n\log n)$ time
%\cite{ref:ChenEf13arXiv,ref:ColeSl87,ref:MegiddoNe83}.

\subsection{Our Techniques}
\label{sec:techniques}

Consider any uncertain point $P_i\in \calP$ and any (deterministic) point $q$ in the
plane $\bbR^2$. We first show that $\Ed(P_i,q)$ is a convex piecewise linear
function with respect to $q\in \bbR^2$.
More specifically, if we extend a horizontal line and a vertical line from
each location of $P_i$, these lines partition the plane into
a grid $G_i$ of $(m+1)\times(m+1)$ cells. Then,
$\Ed(P_i,q)$ is a linear function (in both the $x$-
and $y$-coordinates of $q$) in each cell of $G_i$. In other words,
$\Ed(P_i,q)$ defines a plane surface patch in 3D on each cell of $G_i$.
Then, finding $q^*\in \bbR^2$ is equivalent to
finding a lowest point $p^*$ in the upper envelope of the $n$ graphs in 3D
defined by $\Ed(P_i,q)$ for all $P_i\in \calP$ (specifically, $q^*$ is the
projection of $p^*$ onto the $xy$-plane).

The problem of finding $p^*$, which may be interesting in its own right,
can be solved in $O(nm^2)$ time by the linear-time algorithm for the 3D linear
programming (LP) problem \cite{ref:MegiddoLi83}. Indeed, for a plane surface
patch, we call the plane containing it the {\em supporting plane}.
Let $\calH$ be the set of the supporting planes of the surface patches of the
functions $\Ed(P_i,q)$ for all $P_i\in \calP$.
Since each function $\Ed(P_i,q)$ is convex, $p^*$ is also a lowest
point in the upper envelope of the planes of $\calH$. Thus, finding $p^*$ is a
LP problem in $\bbR^3$ and can be solved in $O(|\calH|)$ time
\cite{ref:MegiddoLi83}.
Note that $|\calH|=\Theta(nm^2)$ since each grid $G_i$ has $(m+1)^2$ cells.

We give an $O(mn)$ time algorithm without computing
the functions $\Ed(P_i,q)$ explicitly. We use a
prune-and-search technique that can be considered as an extension of
Megiddo's technique for the 3D LP problem \cite{ref:MegiddoLi83}.
In each recursive step, we prune at
least $n/32$ uncertain points from $\calP$ in linear time.
In this way, $q^*$ can be found after $O(\log n)$ recursive steps.
%In this way, we can prune all but a constant number of functions of $E$ in
%overall $O(mn)$ time such that the lowest point $w^*$ is determined by
%the remaining $O(1)$ functions. Then, $w^*$ can be found in additional
%$O(m)$ time.

%We should point out that there are some essential
%difficulties to apply Megiddo's technique directly to obtain an
%$O(mn)$ time algorithm. For example,
%Megiddo's pruning scheme relies on the following observation. Suppose
%we are looking for the lowest point $w$ of the upper envelope of a set
%$S$ of planes in $\bbR^2$. Consider
%two planes of $S$ that intersect. In general, their intersection
%is a line and we project
%the line on the $xy$-plane to obtain another line $l$. The line $l$
%partitions the $xy$-plane into two halves. If we know which half of
%the plane contains the projection of $w$ on the $xy$-plane,
%then one of the above two planes can be pruned (i.e, it is useless for
%determining $w$). In our problem, however, the intersection of two
%functions of $E$ is in general not a line but a piecewise linear curve of $\Theta(m)$
%segments. This prevents us from applying Megiddo's approach.
%Instead, we use the following approach.

Unlike Megiddo's algorithm \cite{ref:MegiddoLi83}, each recursive step
of our algorithm itself is a recursive algorithm of $O(\log m)$ recursive steps.
Therefore, our algorithm has $O(\log n)$ ``outer'' recursive steps and
each outer recursive step has $O(\log m)$ ``inner'' recursive steps.
In each outer recursive step, we maintain a rectangle $R$ that always
contains $q^*$ in the $xy$-plane. Initially,
$R$ is the entire plane. Each inner recursive step shrinks $R$ with
the help of a {\em decision algorithm}. The key idea is that after $O(\log m)$
steps, $R$ is so small that there is a set $\calP^*$ of at least $n/2$ uncertain points
such that $R$ is contained inside a single cell of the grid $G_i$ of each
uncertain point $P_i$ of $\calP^*$ (i.e., $R$ does not intersect the
extension lines from the locations of $P_i$). At this point, with the
help of our decision algorithm, we can use a pruning procedure similar
to Megiddo's algorithm \cite{ref:MegiddoLi83} to prune at least
$|\calP^*|/16\geq n/32$ uncertain points of $\calP^*$. Each outer
recursive step is carefully implemented so that it takes only linear time.

In particular, our decision algorithm is for the following {\em decision
problem}. Let $R$ be a rectangle in the plane and $R$ contains $q^*$
(but the exact location of $q^*$ is unknown).
Given an arbitrary line $l$ that intersects
$R$, the decision problem is to determine which side of $l$
contains $q^*$. Megiddo's technique \cite{ref:MegiddoLi83} gave an
algorithm that can solve our decision problem in $O(m^2 n)$
time. We give a decision algorithm of $O(mn)$ time.
In fact, in order to achieve the overall $O(mn)$ time for computing
$q^*$, our decision algorithm has the following performance.
For each $1\leq i\leq n$, let $a_i$ and $b_i$ be
the number of columns and rows of the grid $G_i$ intersecting $R$,
respectively.
%Let $\alpha=\sum_{P_i\in\calP} a_i$ and
%$\beta=\sum_{P_i\in \calP} b_i$. %With $O(mn)$ time preprocessing,
Our decision algorithm runs in  $O(\sum_{i=1}^n (a_i+b_i))$ time.

The rest of the paper is organized as follows. In Section
\ref{sec:obser}, we introduce some observations.
In Section \ref{sec:decision}, we present our decision
algorithm. Section \ref{sec:algo} gives the overall algorithm for
computing the rectilinear center $q^*$. Section \ref{sec:conclude} concludes.

\section{Observations}
\label{sec:obser}
Let $p$ be a point in the plane $\bbR^2$. The vertical line and the horizontal line through $p$ partition the plane into four (unbounded) rectangles.
Consider another point $q\in \bbR^2$. We consider $d(p,q)$ as a function of
$q\in \bbR^2$. For each of the above rectangle $R$, $d(p,q)$ on $q\in R$
is a linear function in both the $x$- and $y$-coordinates of $q$, and thus
$d(p,q)$ on $q\in R$ defines a plane surface patch in $\bbR^3$.
Further, $d(p,q)$ on $q\in \bbR^2$ is a convex piecewise linear
function.

For ease of exposition, we make a general position assumption that no two
locations of the uncertain points of $\calP$ have the same $x$- or
$y$-coordinate.

Consider an uncertain point $P_i$ of $\calP$. We extend a horizontal
line and a vertical line through each location of $P_i$ to obtain a
grid, denoted by $G_i$, which has $(m+1)\times(m+1)$ cells (and each cell is a
rectangle).
%Note that $G_i$ has $(m+1)\times(m+1)$ (possibly unbounded) rectangles
%and we let each rectangle contains its boundary.
According to the above discussion, for each location $p_{ij}$ of $\calP$, the function
$d(p_{ij},q)$ of $q$ in each cell of $G_i$ is linear and
defines a plane surface patch in $\bbR^3$. Therefore, if we consider
$\Ed(P_i,q)$ as a function of $q$, since $\Ed(P_i,q)$ is the sum of
$f_{ij}\cdot d(p_{ij},q)$ for all $1\leq j\leq m$, $\Ed(P_i,q)$ of $q$
in each cell of $G_i$ is also linear and defines a plane surface patch in $\bbR^3$.
Further, since each $d(p_{ij},q)$ for $q\in \bbR^2$ is
convex, the function $\Ed(P_i,q)$, as the sum of convex functions,
is also convex.

In the following, since $\Ed(P_i,q)$ is normally considered as function of $q$,
for convenience, we will use $\Ed_i(x,y)$ to denote it for
$q=(x,y)\in \bbR^2$.

The above discussion leads to the following observation.

\begin{observation}\label{obser:10}
For each uncertain point $P_i\in \calP$,
the function $\Ed_i(x,y)$ is convex piecewise linear.
More specifically,  $\Ed_i(x,y)$ on each cell of the grid $G_i$
is linear and defines a plane surface patch in $\bbR^3$ (e.g.,
see Fig. ~\ref{fig:Edfunction}).
\end{observation}

Consider the function $\Ed_i(x,y)$ of any $P_i\in \calP$. Clearly, the
complexity of $\Ed_i(x,y)$ is $\Theta(m^2)$. However, since $\Ed_i(x,y)$ on each
cell $C$ of $G_i$ is a plane surface patch in $\bbR^3$, $\Ed_i(x,y)$ on $C$ is
of constant complexity. We use $\Ed_i(x,y,C)$ to denote the linear function of
$\Ed_i(x,y)$ on $C$. Note that $\Ed_i(x,y,C)$ is also the function of the
supporting plane of the surface patch of $\Ed_i(x,y)$ on $C$.

As discussed in Section \ref{sec:techniques}, our algorithm will not
compute the function $\Ed_i(x,y)$ explicitly. Instead, we will
compute it implicitly. More specifically, we will do some
preprocessing such that given any cell $C$ of $G_i$, the function $\Ed_i(x,y,C)$
can be determined efficiently. We first introduce some notation.

\begin{figure}[t]
\begin{minipage}[t]{\linewidth}
\begin{center}
\includegraphics[totalheight=2.0in]{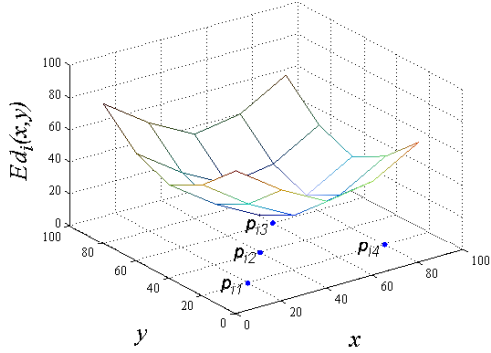}
\caption{\footnotesize Illustrating the function $\Ed_i(x,y)$ of an uncertain point $P_i$ with $m=4$.}
\label{fig:Edfunction}
\end{center}
\end{minipage}
\vspace*{-0.15in}
\end{figure}

Let $X_i=\{x_{i1}, x_{i2}, \cdots, x_{im}\}$ be the set of the
$x$-coordinates of all locations of $P_i$ sorted in ascending order.
Let $Y_i=\{y_{i1}, y_{i2}, \cdots, y_{im}\}$ be the set of
their $y$-coordinates in ascending order.
Note that $X_i$ and $Y_i$ can be obtained in $O(m)$ time
from the input (recall that the locations of $P_i$ are given in two
sorted lists in the input).
For convenience of discussion, we let $x_{i0}=-\infty$,
and let $X_i$ also include $x_{i0}$. Similarly, let
$y_{i0}=-\infty$, and let $Y_i$ also include $y_{i0}$. Note that due
to our general position assumption, the values in $X_i$ (resp., $Y_i$)
are distinct.

For any value $z$, we refer to the largest value in $X_i$ that is
smaller or equal to $z$ the {\em predecessor} of $z$ in $X_i$, and we
use $I_z(X_i)$ to denote the index of the predecessor. Similarly,
$I_z(Y_i)$ is the index of the predecessor of $z$ in $Y_i$.

Consider any point $q$ in the plane. The predecessor of the $x$-coordinate
of $q$ in $X_i$ is also called the {\em predecessor of $q$} in $X_i$. Similarly, the predecessor of the $y$-coordinate of $q$ in $Y_i$ is also called the {\em predecessor of $q$} in $Y_i$.
We use $I_q(X_i)$ and $I_q(Y_i)$ to denote their indices, respectively.

Consider any cell $C$ of the grid $G_i$. For convenience of discussion, we assume $C$
contains its left and bottom sides, but does not contain its top and right
sides. In this way, any point in the plane is contained in one and only one cell
of $G_i$. Further, all points of $C$ have the same predecessor in $X_i$ and also
have the same predecessor in $Y_i$. This allows us to define the {\em predecessor
of $C$} in $X_i$ as the predecessor of any point in $X_i$, and we use $I_C(X_i)$
to denote the index of the predecessor. We define $I_C(Y_i)$ similarly.
We have the following lemma.

\begin{lemma}\label{lem:10}
For any uncertain point $P_i\in \calP$, after $O(m)$ time preprocessing, for any
cell $C$ of the grid $G_i$, if $I_C(X_i)$ and $I_C(Y_i)$ are known, then
the function $\Ed_i(x,y,C)$ can be computed in constant time.
\end{lemma}
\begin{proof}
For each location $p\in P_i$, let $x_p$ and $y_q$ be the $x$- and $y$-coordinates of $p$, respectively, and let $f_p$ be the probability associated with $p$.

For any point $q=(x,y)$ in $\bbR^2$, recall that the expected distance
function $\Ed_i(x,y)=\sum_{p\in P_i}f_p\cdot d(p,q)=\sum_{p\in
P_i}f_p\cdot (|x_p-x|+|y_p-y|)$. Therefore, we can write
$\Ed_i(x,y)=\sum_{p\in P_i}f_p\cdot |x_p-x|+\sum_{p\in P_i}f_p\cdot
|y_p-y|$. In the following, we first discuss how to compute
$\sum_{p\in P_i}f_p\cdot |x_p-x|$ and the case for $\sum_{p\in
P_i}f_p\cdot |y_p-y|$ is very similar.

Let $S_1$ denote the set of all locations of $P_i$ whose
$x$-coordinates are smaller than or equal to $x$, i.e., the
$x$-coordinate of $q$. Let
$S_2=P_i\setminus S_1$. Then, we have the following:
\begin{equation}
\label{equ:10}
\begin{split}
\sum_{p\in P_i}f_p\cdot |x_p-x| &= \sum_{p\in S_1}f_p\cdot (x-x_p) + \sum_{p\in S_2}f_p\cdot (x_p-x) \\
& = x\cdot \big(\sum_{p\in S_1}f_p-\sum_{p\in S_2}f_p \big) - \sum_{p\in S_1}f_p\cdot x_p + \sum_{p\in S_2}f_p\cdot x_p\\
& = x\cdot \big(2\cdot \sum_{p\in S_1}f_p-\sum_{p\in P_i}f_p\big) - 2\sum_{p\in S_1}f_p\cdot x_p + \sum_{p\in P_i}f_p\cdot x_p.
\end{split}
\end{equation}

Thus, in order to compute $\sum_{p\in P_i}f_p\cdot |x_p-x|$, it is sufficient to know the four values $\sum_{p\in S_1}f_p$, $\sum_{p\in P_i}f_p$, $\sum_{p\in S_1}f_p\cdot x_p$, and $\sum_{p\in P_i}f_p\cdot x_p$. To this end, we do the following preprocessing.

First, we compute $\sum_{p\in P_i}f_p$ and $\sum_{p\in P_i}f_p\cdot x_p$, which can be done in $O(m)$ time. Second, recall that $X_i=\{x_{i0},x_{i1},\ldots, x_{im}\}$ maintains the $x$-coordinates of the locations of $P_i$ sorted in ascending order. Note that given any index $j$ with $1\leq j\leq m$, we can access the information of the location of $P_i$ whose $x$-coordinate is $x_{ij}$ in constant time, and this can be done by linking each $x_{ij}$ to the corresponding location of $P_i$ when we create the list $X_i$ from the input. For each $j$ with $1\leq j\leq n$, we let $f(j)$ be the probability associated with the location of $P_i$ whose $x$-coordinate is $x_{ij}$.

In the preprocessing, we compute two arrays $A[0\cdots m]$ and $B[0\cdots m]$. Specifically, for each $1\leq j\leq n$, $A[j]=\sum_{k=1}^j f(k)$ and $B[j]=\sum_{k=1}^jf(k)\cdot x_{ik}$. For $j=0$, we let $A[0]=B[0]=0$.
As discussed above, since  we can access $f(j)$ in constant time for any $1\leq j\leq m$, the two arrays $A$ and $B$ can be computed in $O(m)$ time.

Let $t=I_q(X_i)$, i.e., the index of the predecessor of $q$ in $X_i$. Note that $t\in [0,m]$. To compute $\sum_{p\in P_i}f_p\cdot |x_p-x|$, an easy observation is that $\sum_{p\in S_1}f_p$ is exactly equal to $A[t]$ and $\sum_{p\in S_1}f_p\cdot x_p$ is exactly equal to $B[t]$. Therefore, with the above preprocessing, if $t$ is known, according to Equation \eqref{equ:10}, $\sum_{p\in P_i}f_p\cdot |x_p-x|$ can be computed in $O(1)$ time.

The above shows that with $O(m)$ time preprocessing, given $I_q(X_i)$, we can compute the function $\sum_{p\in P_i}f_p\cdot |x_p-x|$ of $x$ at $q=(x,y)$ in constant time.

In a similar way, with $O(m)$ time preprocessing, given $I_q(Y_i)$, we can compute the function $\sum_{p\in P_i}f_p\cdot |y_p-y|$ of $y$ at $p=(x,y)$ in constant time.

Let $q$ be any point in the cell $C$. Hence, $I_q(X_i)=I_C(X_i)$ and $I_q(Y_i)=I_C(Y_i)$. Further, the function $\Ed_i(x,y)$ on $q=(x,y)\in C$ is exactly the function $\Ed_i(x,y,C)$.
Therefore, with $O(m)$ time preprocessing, given $I_C(X_i)$ and $I_C(Y_i)$, we can compute the function $\Ed_i(x,y,C)$ in constant time.

The lemma thus follows.
\qed
\end{proof}

%\begin{observation}\label{obser:1}
%Given any point $q$ in the $x,y$-plane, after $O(m)$ preprocessing,
%the formula of $\Ed_i(x,y)$ over the rectangles of $A_i$ containing $q$
%can be determined in $O(|\calX_i|+|\calY_i|)$ time, i.e., $O(m)$
%and then the expected rectilinear distance
%$\Ed_i(x_q,y_q)$ of $P_i$ to $q$ will be known in constant time.
%\end{observation}

Due to Lemma \ref{lem:10}, we have the following corollary.

\begin{corollary}\label{coro:10}
For each uncertain point $P_i\in \calP$,
after $O(m)$ time preprocessing, given any point $q$ in the plane,
the expected distance $\Ed(P_i,q)$ can be computed in $O(\log m)$ time.
\end{corollary}
\begin{proof}
Given any point $q\in \bbR^2$, we can compute $I_q(X_i)$ in $O(\log m)$ time by doing
binary search on $X_i$. Similarly, we can compute $I_q(Y_i)$ in $O(\log m)$ time.
Let $C$ be the cell containing $q$. Recall that $I_C(X_i)=I_q(X_i)$ and
$I_C(Y_i)=I_q(Y_i)$. Hence, by Lemma \ref{lem:10}, we can compute the function
$\Ed_i(x,y,C)$ in constant time. Then, $\Ed(P_i,q)$ is equal to
$\Ed_i(q_x,q_y,C)$, where $q_x$ and $q_y$ are the $x$- and $y$-coordinates of
$q$, respectively. Thus, after $\Ed_i(x,y,C)$ is known, $\Ed(P_i,q)$ can be
computed in constant time.  The corollary thus follows. \qed
\end{proof}

Recall that $\Ed_{\max}(q)=\max_{P_i\in \calP}\Ed(P_i,q)$ for any point $q$ in
the plane. For convenience,
we use $\Ed_{\max}(x,y)$ to represent $\Ed_{\max}(q)$ as a function of $q=(x,y)\in
\bbR^2$. Note that $\Ed_{\max}(x,y)$  is the upper envelope of the functions
$\Ed_i(x,y)$ for all $i=1,2,\ldots, n$. Since each $\Ed_i(x,y)$ is convex on
$\bbR^2$, $\Ed_{\max}(x,y)$ is also convex on $\bbR^2$. Further, the rectilinear
center $q^*$ corresponds to a lowest point $p^*$ on $\Ed_{\max}(x,y)$.
Specifically, $q^*$ is the projection of $p^*$ on the $xy$-plane. Therefore,
computing $q^*$ is equivalent to computing a lowest point in the upper envelope
of all functions $\Ed_i(x,y)$ for all $i=1,2,\ldots, n$.

For each $1\leq i\leq n$, let $H_i$ denote the set of supporting planes of all
surface patches of the function $\Ed_i(x,y)$. Let $\calH=\cup_{i=1}^n H_i$.
Since $\Ed_i(x,y)$ is convex,
$\Ed_i(x,y)$ is essentially the upper envelope of the planes in
$H_i$. Hence, $\Ed_{max}(x,y)$ is also the upper envelope of all planes in
$\calH$. Therefore, as discussed in Section \ref{sec:techniques}, finding $p^*$
is essentially a 3D LP problem on $\calH$, which can be solved in $O(|\calH|)$
time by Megiddo's technique \cite{ref:MegiddoLi83}. Since the size of each $H_i$
is $\Theta(m^2)$, $|\calH|=\Theta(nm^2)$.
Therefore, applying the algorithm in \cite{ref:MegiddoLi83} directly can solve
the problem in $O(nm^2)$ time.
In the following, we give an $O(nm)$ time algorithm.

In the following paper, we assume we have done the preprocessing of
Lemma \ref{lem:10} for each $P_i\in \calP$, which takes $O(mn)$ time in
total.
%We first give a decision algorithm in Section \ref{sec:decision}

%As introduced in Section \ref{Sec:1.2}, the center $q^*$ can be computed in $O(nm^2)$
%by the linear programming algorithm in $\bbR^3$ \cite{ref:MegiddoLi83}.
%However, our algorithm presented as follows, can solve it in $O(mn)$ time.
%To compute $q^*$, we will first propose the {\em decision} algorithm
%that can recognize in which half-plane determined
%by a given straight line in $x,y$-plane $q^*$ lies.
%
%Note that the time complexity of our decision algorithm mainly
%depends on the sizes of $\calX_i$ and $\calY_i$ for all $1\leq i\leq n$.
%For convenience, we let $h_i$ represent the total size of $\calX_i$ and $\calY_i$,
%i.e., the value $|\calX_i|+|\calY_i|$,
%and denote the sum of $h_i$ for all $1\leq i\leq n$ by $H$.
%Each $h_i$ is initially bounded by $O(m)$ and hence $H$ is bounded by $O(mn)$.

\section{The Decision Algorithm}
\label{sec:decision}
In this section, we present a decision algorithm that solves a decision
problem, which is needed later in Section \ref{sec:algo}. We first introduce
the decision problem.

Let $R=[x_1,x_2;y_1,y_2]$ be an axis-parallel rectangle in the plane, where
$x_1$ and $x_2$ are the $x$-coordinates of the left and right sides of $R$, respectively,
and $y_1$ and $y_2$ are the $y$-coordinates of the bottom and top sides of $R$,
respectively.
Suppose it is known that $q^*$  is in $R$ (but the exact location of $q^*$ is
not known). Let $L$ be an arbitrary line that
intersects the interior of $R$. The {\em decision problem} asks whether $q^*$ is on
$L$, and if not, which side of $L$ contains $q^*$.
We assume the two predecessor indices $I_{x_1}(X_i)$ and $I_{y_1}(Y_i)$ are
already known.

For each $1\leq i\leq n$, let
$a_i=I_{x_2}(X_i)-I_{x_1}(X_i)+1$ and $b_i=I_{y_2}(Y_i)-I_{y_1}(Y_i)+1$.
In fact, $a_i$ and $b_i$ are the numbers of columns and rows of
$G_i$ that intersect $R$, respectively. Below, we give a {\em decision
algorithm} that solves the decision problem in $O(\sum_{i=1}^n
(a_i+b_i))$ time. Note that $2n \leq \sum_{i=1}^n (a_i+b_i)\leq 2(m+1)n$.

We first show that the decision problem can be solved in
$O(\sum_{i=1}^na_i\cdot b_i)$ time by using the decision algorithm
for the 3D LP problem \cite{ref:MegiddoLi83}. Later we
will reduce the running time to $O(\sum_{i=1}^n
(a_i+b_i))$ time.

Recall that $p^*$ is a lowest point in the upper envelope of the functions
$\Ed_i(x,y)$ for $i=1,2,\ldots,n$.
Since $q^*$ is in $R$ and each function $\Ed_i(x,y)$ is convex, an easy
observation is that $p^*$ is also a
lowest point in the upper envelope of $\Ed_i(x,y)$ for
$i=1,2,\ldots,n$ restricted on $(x,y)\in R$. This implies that we only need to
consider each function $\Ed_i(x,y)$ restricted on $R$.

For each $1\leq i\leq n$, let $G_i(R)$ be the set of cells of $G_i$ that
intersect $R$, and let $H_i(R)$ be the set of
supporting planes of the surface patches of
$\Ed(P_i,q)$ defined on the cells of $G_i(R)$. Let $\calH(R)=\cup_{i=1}^n
H_i(R)$. By our above analysis,
$p^*$ is a lowest point of the upper envelope of all planes in
$\calH(R)$. Note that $|H_i(R)|=a_i\cdot b_i$ for each $1\leq i\leq n$.
Thus, $|\calH(R)|=\sum_{i=1}^n a_i\cdot b_i$.
Then, we can apply the decision algorithm in \cite{ref:MegiddoLi83}
(Section 5.2)
%(which we call the {\em 3D LP decision algorithm})
on $\calH(R)$ to determine which side of $L$ contains $q^*$ in
$O(|\calH(R)|)$ time. In order to explain our improved algorithm later, we
sketch this algorithm below.

We consider each plane of $\calH(R)$ as a function of the points $q$ on the
$xy$-plane $\bbR^2$. In the first step, the algorithm finds a point $q'$ on $L$ that minimizes
the maximum value of all functions in $\calH(R)$ restricted on the line $q\in
L$. This is essentially a 2D LP problem because each function of
$\calH(R)$ restricted on $L$ is a line, and thus the
problem can be solved in $O(|\calH(R)|)$ time \cite{ref:MegiddoLi83}.
Let $\Phi_{q'}$ be the set of
functions of $\calH(R)$ whose values at $q'$ are equal to the above maximum value.
The set $\Phi_{q'}$ can be found in $O(|\calH(R)|)$ time after $q'$ is
computed. This
finishes the first step, which takes $O(|\calH(R)|)=O(\sum_{i=1}^na_i\cdot b_i)$ time.

The second step solves another two instances of the 2D LP problem on the
planes of $\Phi_{q'}$, which takes $O(|\Phi_{q'}|)$ time. An easy
upper bound for $|\Phi_{q'}|$ is $\sum_{i=1}^na_i\cdot b_i$.
A close analysis can show that $|\Phi_{q'}|=O(n)$.
Indeed, for each $1\leq i\leq n$, since the function $\Ed_i(x,y)$ is convex, among all
$a_i\cdot b_i$ planes in $H_i(R)$, at most four of them are in $\Phi_{q'}$. Therefore,
$|\Phi_{q'}|=O(n)$. Hence, the second step runs in $O(n)$ time.
Since in our problem there always exists a solution,
according to \cite{ref:MegiddoLi83}, the second step will either conclude
that $q'$ is $q^*$ or tell which side of $L$ contains $q^*$, which
solves the decision problem. The algorithm takes
$O(\sum_{i=1}^na_i\cdot b_i)$ time in total, which is
dominated by the first step.

In the sequel, we reduce the running time of the above algorithm, in
particular, the first step, to $O(\sum_{i=1}^n(a_i+ b_i))$.
Our goal is to compute $q'$ and $\Phi_{q'}$.
By the definition, $q'$ is a
lowest point in the upper envelope of all functions of $\calH(R)$
restricted on the line $L$. Consider any uncertain point $P_i\in \calP$.
Let $H_i(R,L)$ be the set of supporting planes of the
surface patches defined on the cells of $G_i(R)$ intersecting $L$.
Observe that since $\Ed_i(x,y)$ is convex, the upper envelope of all the
functions of $H_i(R)$ restricted on $L$ is exactly the upper envelope of
the functions of $H_i(R,L)$ restricted on $L$. Therefore, $q'$ is also a lowest
point in the upper envelope of the functions of $\calH(R,L)$ restricted on $L$,
where $\calH(R,L)=\cup_{i=1}^nH_i(R,L)$.  In other words, among all planes in
$\calH(R)$, only the planes of $\calH(R,L)$ are relevant for determining $q'$.
Thus, suppose $\calH(R,L)$ has been computed; then
$q'$ can be computed based on the planes of $\calH(R,L)$ in $O(|\calH(R,L)|)$
time by the 2D LP algorithm \cite{ref:MegiddoLi83}.
After $q'$ is computed, the set $\Phi_{q'}$ can also be determined in $O(|\calH(R,L)|)$
time.

Note that $|\calH(R,L)|=O(\sum_{i=1}^n(a_i+b_i))$, since
for each $1\leq i\leq n$, $|H_i(R,L)|$, which
is equal to the number of cells of $G_i(R)$ intersecting
$L$, is $O(a_i+b_i)$.

It remains to compute $\calH(R,L)$, i.e., compute
$H_i(R,L)$ for each $1\leq i\leq n$. Recall that $R=[x_1,x_2;y_1,y_2]$ and the
two predecessor indices $I_{x_1}(X_i)$ and $I_{y_1}(Y_i)$ for each $1\leq i\leq
n$ are already known.  The following lemma gives an
$O(a_i+b_i)$ algorithm to compute $H_i(R,L)$.

\begin{lemma}\label{lem:20}
For each $1\leq i\leq n$, $H_i(R,L)$ can be computed in
$O(a_i+b_i)$ time.
\end{lemma}
\begin{proof}
Let $G_i(R,L)$ be the set of cells of $G_i(R)$ intersecting $L$.
To compute the planes in $H_i(R,L)$, it is sufficient to determine the plane
surface patches of $\Ed_i(x,y)$ defined on the cells of $G_i(R,L)$.
By Lemma \ref{lem:10}, this amounts to determine the indices of the
predecessors of these cells in $X_i$ and $Y_i$, respectively.
In the following, we give an algorithm to compute the cells of $G_i(R,L)$ and
determine their predecessor indices in $X_i$ and $Y_i$, respectively, and the
algorithm runs in $O(a_i+b_i)$ time.

The main idea is that we first pick a particular
point $p$ on $L\cap R$ and locate the cell of
$G_i(R)$ containing $p$ (clearly this cell belongs to $G_i(R,L)$),
and then starting from $p$, we traverse on $L$ and $G_i(R)$
simultaneously to trace other cells of $G_i(R,L)$ until we move out of $R$. The
details are given below.

We focus on the case where $L$ has a positive slope. The other cases can be
handled similarly.
%Without loss of generality, we assume $L$ is vertical or has a positive slope.
Recall that $L$ intersects the interior of $R$. %If $L$ has a positive slope,
Let $p$ be the leftmost intersection of $L$ with the boundary of $R$. Hence, $p$
is either on the left side or the bottom side of $R$.
%Below, we assume $p$ is the on
%the left side of $R$ and the other case can be handled similarly.

Let $C$ be the cell of $G_i$ that contains $p$. We first determine
the two indices $I_p(X_i)$ and $I_p(Y_i)$. Note that
$I_C(X_i)=I_p(X_i)$ and $I_C(Y_i)=I_p(Y_i)$.

Since $p\in R$, the index $I_p(X_i)$ can be found in $O(a_i)$ time by scanning the
list $X_i$ from the index $I_{x_1}(X_i)$.
Similarly, $I_p(Y_i)$ can be found in $O(b_i)$ time by scanning the
list $Y_i$ from the index $I_{y_1}(Y_i)$.
After $I_C(X_i)=I_p(X_i)$ and $I_C(Y_i)=I_p(Y_i)$ are computed,
by Lemma \ref{lem:10}, the function $\Ed_i(x,y,C)$ can be computed in
constant time, and we add the function to $H_i(R,L)$.

Next, we move $p$ on $L$ rightwards. We will show that when $p$ crosses the
boundary of $C$, we can determine the new cell containing $p$ and update the
two indices $I_p(X_i)$ and $I_p(Y_i)$ in constant time. This process continues
until $p$ moves out of $R$.  Specifically, when $p$ moves on $L$
rightwards, $p$ will cross the boundary of $C$ either from the top side or the
right side.

First, we determine whether $p$ will move out of $R$
before $p$ crosses the boundary of $C$. If yes, then we terminate the
algorithm. Otherwise, we determine whether $p$ moves out of $C$ from
its right side or left side. All above can be easily done in constant
time. Depending on whether $p$ crosses the boundary of $C$ from its top side,
right side, or from both sides simultaneously, there are three cases.

%To determine which side, we use the following approach.
%Note the right side of $C_p$ is contained
%in the vertical line $x=x_{i,i(p,x)+1}$ and top side of $C_p$ is contained in the
%horizontal line $y=y_{i,i(p,y)+1}$. We compute the intersection $p_v$ of $L$
%with the above vertical line and compute the intersection $p_h$ of $l$ with the
%above horizontal line. Depending on which point is closer to $p$, there are
%three cases.

\begin{enumerate}
\item
If $p$ crosses the boundary of $C$ from the top side and $p$ does not
cross the right side of $C$, then $p$ enters into a new cell that is on top of $C$.
We update $C$ to the new cell. We increase the index $I_p(Y_i)$ by one, but
keep $I_p(X_i)$ unchanged. Clearly, the above two indices are
correctly updated and $I_C(X_i)=I_p(X_i)$ and $I_C(Y_i)=I_p(Y_i)$ for
the new cell $C$. Again, by Lemma \ref{lem:10}, the function
$\Ed_i(x,y,C)$ for the new cell $C$ can be
computed in constant time. We add the new function to $H_i(R,L)$.

\item
If $p$ crosses the boundary of $C$ from the right side and $p$ does not
cross the top side of $C$, then $p$ enters into a new cell that is on
right of $C$. The algorithm in this case is similar to the above case and we omit the
discussions.

\item
The remaining case is when $p$ crosses the boundary of $C$ through the
top right corner of $C$. In this case, $p$ enters into the northeast
neighboring cell of $C$. We first add to $H_i(R,L)$ the supporting planes of the
surface patches of $\Ed_i(x,y)$ defined on the top neighboring cell and the right neighboring
cell of $C$, which can be computed in constant time as
the above two cases. Then, we update $C$ to the new cell $p$ is entering.
We increase each of $I_p(X_i)$ and $I_p(Y_i)$ by one. Again, the two
indices are correctly updated for the new cell $C$.
Finally, we compute the new function $\Ed_i(x,y,C)$ and add it to $H_i(R,L)$.
\end{enumerate}

When the algorithm stops, $H_i(R,L)$ is computed.
In general, during the procedure of moving $p$ on $L$, we spend constant time on
finding each supporting plane of $H_i(R,L)$.  Therefore, the total running time
of the entire algorithm is $O(a_i+b_i)$. The lemma thus follows.
\qed
\end{proof}

With the preceding lemma, we have the following result.

\begin{theorem}\label{theo:10}
The decision problem can be solved in $O(\sum_{i=1}^n(a_i+b_i))$ time.
\end{theorem}

\section{Computing the Rectilinear Center}
\label{sec:algo}

In this section, with the help of our decision algorithm in Section
\ref{sec:decision}, we compute the rectilinear center $q^*$ in $O(mn)$ time.

As discussed in Section \ref{sec:techniques}, our algorithm is a
prune-and-search algorithm that has $O(\log n)$ ``outer''
recursive steps each of which has $O(\log m)$ ``inner'' recursive
steps. In each outer recursive step, the algorithm prunes at
least $|\calP|/32$ uncertain points of $\calP$ such that these uncertain
points are not relevant for computing $q^*$. After
$O(\log n)$ outer recursive steps, there will be only a constant number of
uncertain points remaining in $\calP$. Each outer recursive step runs in
$O(m |\calP|)$ time,
where $|\calP|$ is the number of uncertain points remaining in $\calP$.
In this way, the total running time of the
algorithm is $O(mn)$.

%we obtain a recurrence $T(m,n)=T(m,31/32\cdot  n) + O(mn)$,
%and solving the recurrence will give $T(m,n)=O(mn)$.

Each outer recursive step is another recursive
prune-and-search algorithm, which consists of $2+\log m$ inner recursive steps.
%Note that this is one major difference from the 3D linear programming
%algorithm in \cite{ref:MegiddoLi83}.
Let $\calX=\cup_{i=1}^n X_i$ and $\calY=\cup_{i=1}^n Y_i$.
Hence, $|\calX|=|\calY|=mn$. We maintain a rectangle
$R=[x_1,x_2;y_1,y_2]$ that contains $q^*$. Initially, $R$ is the entire plane. In
each inner recursive step, we shrink $R$ such that the {\em $x$-range} $[x_1,x_2]$
(resp., {\em $y$-range} $[y_1,y_2]$) of the new $R$ only contains half of the values
of $\calX$ (resp., $\calY$) in the $x$-range (resp., $y$-range) of the previous
$R$. In this way, after $\log m+2$
inner recursive steps, the $x$-range (resp., $y$-range) of $R$ only contains at most $n/4$ values of $\calX$ (resp., $\calY$). At this
moment, a {\em key observation} is that there is a subset $\calP^*$ of
at least $n/2$ uncertain points, such that for each $P_i\in \calP^*$, $R$ is contained in the interior of a cell of the grid $G_i$, i.e., the $x$-range (resp., $y$-range) of $R$ does not contain any value of $X_i$ (resp., $Y_i$). Due to the observation,
we can use a pruning procedure similar to that in
\cite{ref:MegiddoLi83} to prune at least $|\calP^*|/16\geq n/32$ uncertain
points.

In the following, in Section \ref{sec:firstprune}, we give our algorithm on pruning the values of $\calX$ and $\calY$ to obtain $\calP^*$. In Section \ref{sec:secondprune}, we prune uncertain points of $\calP^*$.

\subsection{Pruning the Coordinate Values of $\calX$ and $\calY$}
\label{sec:firstprune}

%Hence, the grid $G$ is formed by the vertical lines with $x$-coordinates in $X_i$
%and horizontal lines with $y$-coordinates in $Y_i$.
%We consider the first outer recursive step on $\calP$.

Consider a general step of the algorithm where we are about to perform the
$j$-th inner recursive step for $1\leq j\leq \log m+2$. Our algorithm maintains
the following {\em algorithm invariants}. (1) We have a rectangle
$R^{j-1}=[x_1^{j-1},x_2^{j-1};y_1^{j-1},y_2^{j-1}]$ that contains $q^*$. (2) For
each $1\leq i\leq n$, the index $I_{x_1^{j-1}}(X_i)$ of the predecessor of
$x_1^{j-1}$ in $X_i$ is known, and so is the index $I_{y_1^{j-1}}(Y_i)$. (3) We
have a sublist $X_i^{j-1}$ of $X_i$ that consists of all values of
$X_i$ in $[x_1^{j-1},x_2^{j-1}]$ and a sublist $Y_i^{j-1}$ of $Y_i$
that consists of all values of $Y_i$ in $[y_1^{j-1},y_2^{j-1}]$. Note that these
sublists can be empty. (4) $|\calX^{j-1}|\leq mn/2^{j-1}$ and
$|\calY^{j-1}|\leq mn/2^{j-1}$, where $\calX^{j-1}=\cup_{i=1}^nX_i^{j-1}$  and
$\calY^{j-1}=\cup_{i=1}^n Y_i^{j-1}$.

Initially, we set $R^0=[-\infty,+\infty;-\infty,+\infty]$, $X_i^0=X_i$ and
$Y_i^0=Y_i$ for each $1\leq i\leq n$, with $\calX^0=\calX$ and $\calY^0=\calY$. It
is easy to see that before we start the first inner recursive step for
$j=1$, all the algorithm invariants hold.

In the sequel, we give the details of the $j$-th inner recursive step. We will
show that its running time is $O(mn/2^{j}+n)$ and all
algorithm invariants are still maintained after the step.

Let $x_m$ be the median of $\calX^{j-1}$ and $y_m$ be the median of $\calY^{j-1}$.
Both $x_m$ and $y_m$ can be found in $O(|\calX^{j-1}|+|\calY^{j-1}|)$ time.

For each $1\leq i\leq n$, let $a_i^{j-1}=I(x_2^{j-1})-I(x_1^{j-1})+1$ and
$b_i^{j-1}=I(y_2^{j-1})-I(y_1^{j-1})+1$. Observe that
$a_i^{j-1}=|X_i^{j-1}|+1$ and $b_i^{j-1}=|Y_i^{j-1}|+1$.

Let $x^*$ and $y^*$ be the $x$- and $y$-coordinates of $q^*$,
respectively.

We first determine whether $x^*>x_m$, $x^*<x_m$, or $x^*=x_m$. This can
be done by applying our decision algorithm on $R^{j-1}$ and
$L$ with $L$ being the vertical line $x=x_m$. By Theorem \ref{theo:10}, the
running time of our decision algorithm is
$O(\sum_{i=1}^n(a_i^{j-1}+b_i^{j-1}))$, which is
$O(n+|\calX^{j-1}|+|\calY^{j-1}|)$.

Note that if $x^*=x_m$, then according to our decision algorithm,
$q^*$ will be found by the decision algorithm
and we can terminate the entire algorithm. Otherwise, without loss of generality, we
assume $x^* > x_m$. We proceed to determine whether $y^*>y_m$ or $y^*<y_m$,
or $y^*=y_m$ by applying our decision algorithm on
$R^{j-1}$ and $L$ with $L$ being the horizontal line $y=y_m$.
Similarly, if $y^*=y_m$, then the decision algorithm will find $q^*$
and we are done. Otherwise, without loss of generality we assume $y^*>y_m$.
The above calls our decision algorithm twice, which takes
$O(n+|\calX^{j-1}|+|\calY^{j-1}|)$ time in total.

Now we know that $q^*$ is in the rectangle
$[x_m,x_2^{j-1};y_m,y_2^{j-1}]$. We let $R^j=[x_1^j,x_2^j;y_1^j,y_2^j]$
be the above rectangle, i.e., $x_1^j=x_m$, $x_2^j=x_2^{j-1}$,
$y_1^j=y_m$, and $y_2^j=y_2^{j-1}$.
Clearly, the first algorithm invariant is maintained.

We further proceed as follows to maintain the other three invariants.

For each $1\leq i\leq n$, by scanning the sorted list $X_i^{j-1}$, we
compute the index $I_{x_1^{j}}(X_i)$ of the predecessor of $x_1^{j}$ in
$X_i$ (each element of $X_i^{j-1}$ maintains its
original index in $X_i$),
and similarly, by scanning the sorted list $Y^{j-1}_i$, we
compute the index $I_{y_1^{j}}(Y_i)$.
Computing these indices in all $X_i$ and $Y_i$ for
$i=1,2,\ldots, n$ can be done in $O(|\calX^{j-1}|+|\calY^{j-1}|)$ time.
%Of course, we can do faster by doing binary search.
%But the linear scan is sufficient for our algorithm.
This maintains the second algorithm invariant.

Next, for each $1\leq i\leq n$, we scan $X_i^{j-1}$ to compute a sublist $X_i^j$,
which consists of all values of $X_i^{j-1}$ in $[x_1^{j},x_2^{j}]$, and similarly,
we scan $Y_i^{j-1}$ to compute a sublist $Y_i^j$, which
consists of all values of $Y_i^{j-1}$ in $[y_1^{j},y_2^{j}]$.
Computing the lists $X_i^j$ and $Y_i^j$ for all $i=1,2,\ldots,n$ as
above can be done in overall $O(|\calX^{j-1}|+|\calY^{j-1}|)$ time.
This maintains the third algorithm invariant.

Let $\calX^j=\sum_{i=1}^n X_i^j$ and $\calY^j=\sum_{i=1}^n Y_i^j$.
According to our above algorithm,
$|\calX^j|\leq |\calX^{j-1}|/2$ and $|\calY^j|\leq |\calY^{j-1}|/2$.
Since $|\calX^{j-1}|\leq nm/2^{j-1}$ and $|\calY^{j-1}|\leq
nm/2^{j-1}$, we obtain $|\calX^{j}|\leq nm/2^{j}$ and $|\calY^{j}|\leq
nm/2^{j}$. Hence, the fourth algorithm invariant is maintained.

In summary, after the $j$-th inner recursive step, all four
algorithm invariants are maintained. Our above analysis also shows that
the total running time is $O(n+|\calX^{j-1}|+|\calY^{j-1}|)$, which is
$O(nm/2^j + n)$.

We stop the algorithm after the $t$-th inner recursive step, for $t=2+\log m$.
The total time for all $t$ steps is thus $O(\sum_{j=1}^t(n+ mn/2^j))=O(mn)$.

After the $t$-th step, by our algorithm invariants, the rectangle $R^t$ contains $q^*$, and $|\calX^t|\leq mn/2^t=n/4$ and $|\calY^t|\leq mn/2^t=n/4$.

We say that an uncertain point $P_i$ is {\em prunable} if both $X_i^t$ and $Y_i^t$ are empty (and thus $R^t$ is contained in the interior of a cell of $G_i$). Let $\calP^*$ denote the set of all prunable uncertain points of $\calP$.
The following is an easy but crucial observation.

\begin{observation}\label{obser:20}
$|\calP^*|\geq n/2$.
\end{observation}
\begin{proof}
Since $\calX^t\leq n/4$, among the $n$ sets $X_i^t$ for $i=1,2,\ldots, n$, at most $n/4$ of them are non-empty. Similarly, since $\calY^t\leq n/4$, among the $n$ sets $Y_i^t$ for $i=1,2,\ldots, n$, at most $n/4$ of them are non-empty. Therefore, there are at most $n/2$ uncertain points $P_i\in \calP$ such that either $X_i^t$ or $Y_i^t$ is non-empty. This implies that there are at least $n/2$ prunable uncertain points in $\calP$.
\qed
\end{proof}

After the $t$-th inner recursive step, the set $\calP^*$ can be obtained in $O(n)$ time by checking all sets $X_i^t$ and $Y_i^t$ for $i=1,2,\ldots, n$ and see whether they are empty.

The reason we are interested in prunable uncertain points is that for each prunable uncertain point $P_i$ of $\calP^*$, since $R^t$ contains $q^*$ and $R_t$ is contained in a cell $C_i$ of $G_i$, there is only one surface patch of $\Ed_i(x,y)$ (i.e., the one defined on $C_i$) that is relevant for computing $q^*$.
%As will be seen in the next subsection, this property will help us to prune uncertain points from $\calP^*$.
%Consider any prunable uncertain point $P_i$. Let $C_i$ denote the cell of $G_i$
%containing $R^t$, and
Let $h_i$ denote the supporting plane of the above surface patch. We call $h_i$ the {\em relevant plane} of $P_i$. Note that we can obtain $h_i$ in constant time.
Indeed, observe that the predecessor index $I_{C_i}(X_i)$ is exactly $I_{x_1^t}(X_i)$, which is known by our algorithm invariants. Similarly, the index $I_{C_i}(Y_i)$ is also known. By Lemma \ref{lem:10}, the function $\Ed_i(x,y,C_i)$, which is also the function of $h_i$, can be obtained in constant time. Hence, the relevant
planes of all prunable uncertain points of $\calP^*$ can be obtained in $O(n)$ time.

\paragraph{Remark.}
One may wonder why we did not perform the inner recursive
steps for $t=\log mn$ times (instead of $t=2+\log m$ time)
so that $\calX^t$ and $\calY^t$ would each have a constant number of values in
the range of $R$.
The reason is that based on our analysis, that would take $O(mn+n\log
nm)$ time, which may not be bounded by $O(mn)$ (e.g., when $m=o(\log n)$). In
fact,  performing the inner recursive steps for $t=2+\log m$ times such that
$\calX^t$ and $\calY^t$ each have at most $\frac{n}{4}$ values in the range of
$R$ is an interesting and crucial ingredient of our techniques.

\subsection{Pruning Uncertain Points from $\calP^*$}
\label{sec:secondprune}

Consider a prunable uncertain point $P_i$ of $\calP^*$. Recall that $H_i$ is the
set of supporting planes of all surface patches of $\Ed_i(x,y)$. The above
analysis shows that among all planes in $H_i$, only the relevant plane
$h_i$ is useful for determining $q^*$.
In other words, the point $p^*$, as a lowest point of all planes in $\calH=\cup_{i=1}^n
H_i$, is also a lowest point of the planes in the union of $\cup_{P_i\in
\calP^*}h_i$ and $\cup_{P_i\in \calP\setminus\calP^*} H_i$.
%Therefore, we can even consider $\Ed(P_i,q)$ as consisting of the only plane $h_i$.
This will allow us to prune at least $|\calP^*|/16$ uncertain points from
$\calP^*$. The idea is similar to  Megiddo's pruning scheme
for the 3D LP algorithm in \cite{ref:MegiddoLi83}.
%We borrow some terminologies from \cite{ref:MegiddoLi83}.

For each $P_i\in\calP^*$, its relevant plane $h_i$ is also considered as a function in
the $xy$-plane.  Arrange all uncertain points of $\calP^*$ into $|\calP^*|/2$
disjoint pairs. Let $D(\calP^*)$ denote the
set of all these pairs. %Since $|\calP'|\geq n/2$, $|D|\geq n/4$.
For each pair $(P_i, P_j)\in D(\calP^*)$, if the value of the function
$h_i$ at any point of $R^t$
is greater than or equal to that of $h_j$, then $P_j$ can be pruned immediately;
otherwise, we project the intersection of $h_i$ and $h_j$
on the $xy$-plane to obtain a line $L_{ij}$ dividing $R^t$ into two parts,
such that $h_i\geq h_j$ on one part and $h_i\leq h_j$ on the other.

Let $\calL$ denote the set of the dividing lines $L_{ij}$ for all pairs of
$D(\calP^*)$. Let $L_m$ be
the line whose slope has the median value among the lines of
$\calL$. We transform the coordinate system by rotating the $x$-axis to be
parallel to $L_m$ (the $y$-axis does not change). For ease of discussion,
we assume no other lines of $\calL$ are parallel to $L_m$ (the assumption
can be easily lifted; see \cite{ref:MegiddoLi83}).  In the new coordinate
system, half the lines of $\calL$ have negative slopes and the other half have
positive slopes. We now arrange all lines of $\calL$ into disjoint pairs
such that each pair has a line of a negative slope and a line of
positive slope. Let $D(\calL)$ denote the set of all these line pairs.
%Clearly, $|D(\calL)|\leq \frac{n}{8}$.
%For simplicity of discussion, we assume every pair
%of lines of $D(\calL)$ are not parallel (the parallel case can be easily
%handled; see \cite{ref:MegiddoLi83}).

%As in Megiddo's algorithm \cite{ref:MegiddoLi83},
%for each pair $(L_i, L_j)\in D(\calL)$, we define a $y$-value $y_{ij}$ as follows:
%if $L_i$ and $L_j$ are parallel, then they must be parallel to the $x$-axis and
%we let $y_{ij}$ be the mean of their $y$-coordinates;
%otherwise, they must intersect and we let $y_{ij}$ be
%the $y$-coordinate of their intersection point.
%Obviously, compute $y_{ij}$ for all pairs of $D(\calL)$ can be done in $O(n)$ time.

For each pair $(L_i, L_j)\in D(\calL)$, we define $y_{ij}$ as the $y$-coordinate
of the intersection of $L_i$ and $L_j$.
We find the median $y_m$ of the values $y_{ij}$ for all pairs in $D(\calL)$.
Let $x^*$ and $y^*$ respectively be the $x$- and $y$-coordinate of $q^*$ in the new
coordinate system.  We determine  in $O(mn)$ time whether $y^*>y_m$,
$y^*<y_m$ or $y^*=y_m$ by using our decision algorithm (here an
$O(mn)$ time decision algorithm is sufficient for our purpose).
If $y^*=y_m$, then our decision algorithm will find $q^*$ and we can terminate the algorithm.
Otherwise, without of loss generality, we assume $y^*<y_m$.

%Consider all pairs of parallel lines with their $y$-value is larger than $y_m$,
%which can be founded in $O(n)$ time.
%Let $(L_i, L_j)$ be such a pair.
%We can see that at least one among $L_i$ and $L_j$
%lies above the line $y=y_m$
%which can be identified in $O(1)$ time.
%Without loss of generality, suppose it is $L_i$.
%Assume $L_i$ is formed by the intersection of
%the functions $\Ed_{i_1}(x,y)$ and $\Ed_{i_2}(x,y)$.
%Definitely, over the whole half-plane $\{(x,y)| y<y_m\}$,
%one of $\Ed_{i_1}(x,y)$ and $\Ed_{i_2}(x,y)$ %lies wholly below another.
%is smaller than another's, which can be recognized in $O(1)$ time.
%Hence, we can eliminate the uncertain point corresponding to
%the smaller of $\Ed_{i_1}(x,y)$ and $\Ed_{i_2}(x,y)$.
%Thus, for each such pair, one uncertain point can be pruned.

Let $D'(\calL)$ denote the set of all pairs $(L_i, L_j)$ of $D(\calL)$
such that $y_{ij}\geq y_m$. Note that $|D'(\calL)|\geq |D(\calL)|/2$.
For each pair $(L_i,L_j)\in D'(\calL)$, let $x_{ij}$ be the $x$-coordinate of the
intersection of $L_i$ and $L_j$.  We find the median $x_m$ of all such $x_{ij}$'s.
By using our decision algorithm, we can determine in $O(mn)$ time
whether $x^*>x_m$, $x^*<x_m$, or $x^*=x_m$.
If $x^*=x_m$, our decision algorithm will find $q^*$ and we
are done. Otherwise, without loss of generality, we assume $x^*<x_m$.

\begin{figure}[t]
\begin{minipage}[t]{\linewidth}
\begin{center}
\includegraphics[totalheight=1.2in]{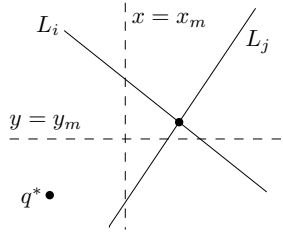}
\caption{\footnotesize The intersection of $L_i$ and $L_j$ is in the first
quarter of the intersection of $x=x_m$ and $y=y_m$ while $q^*$ is in the
interior of the third quarter.}
\label{fig:prune}
\end{center}
\end{minipage}
\vspace*{-0.15in}
\end{figure}

Now for each pair $(L_i, L_j)$ of $D'(\calL)$ with $x_{ij}\geq x_m$ and
$y_{ij}\geq y_m$ (there are at least $|D'(\calL)|/2$ such pairs),
we can prune either $P_i$ or $P_j$, as follows. Indeed, one of the lines in
such a pair $(L_i, L_j)$, say $L_i$, has a negative slope and does
not intersect the region $R=\{(x,y)\ |\ x< x_m,\ y< y_m\}$ (e.g., see
Fig.~\ref{fig:prune}). Suppose $L_i$ is the dividing line of two relevant planes $h_{k_1}$
and $h_{k_2}$ of two uncertain points $P_{k_1}$ and $P_{k_2}$ of $\calP^*$.
It follows that either $h_{k_1}\geq h_{k_2}$ or $h_{k_1}\leq h_{k_2}$ holds
on the region $R$. Since $q^*\in R$, one of $P_{k_1}$ and $P_{k_2}$ can be pruned.

As a summary, the above pruning algorithm prunes at least
$|\calP^*|/16\geq n/32$ uncertain points and the total time is $O(mn)$.

\subsection{Wrapping Things Up}
\label{sec:wrapup}

The algorithm in the above two subsections either computes $q^*$ or prunes at least $n/32$ uncertain points from $\calP$ in $O(mn)$ time. We assume the latter case happens. Then we apply
the same algorithm recursively on the remaining uncertain points for
$O(\log n)$ steps, after which only a constant number of uncertain points
remain. The total running time can be
described by the following recurrence: $T(m,n)=T(m,\frac{31\cdot
n}{32})+O(mn)$. Solving the recurrence gives $T(m,n)=O(mn)$.

Let $\calP'$ be the set of the remaining uncertain points, with
$|\calP'|=O(1)$. Hence, the rectilinear center $q^*$ is determined by $\calP'$.
In other words, $q^*$ is also a rectilinear center of $\calP'$.
In fact, like other standard prune-and-search algorithms, the way we prune uncertain points of $\calP$ guarantees that any rectilinear center of $\calP$ is also a rectilinear center of $\calP'$, and vice versa.
%In other words, a rectilinear center of $\calP$ is also a rectilinear center of $\calP'$, and vice versa.
By using an approach similar to that in Section \ref{sec:firstprune},
Lemma \ref{lem:30} finally computes $q^*$ based on $\calP'$ in $O(m)$
time.

\begin{lemma}\label{lem:30}
The rectilinear center $q^*$ can be computed in $O(m)$ time.
\end{lemma}
\begin{proof}
%The algorithm is similar to that in Section \ref{sec:firstprune}.
Let $c=|\calP'|$, which is a constant. Let $\calX'=\cup_{P_i\in \calP'}X_i$ and
$\calY'=\cup_{P_i\in \calP'}Y_i$. We apply the same recursive
algorithm in Section \ref{sec:firstprune} on $\calX'$ and $\calY'$ for
$O(\log m)$ steps, after which we will obtain a rectangle $R$ such
that $R$ contains $q^*$ and for each $P_i\in \calP'$, the $x$-range
(resp., $y$-range) of $R$ only contains a constant number of values of
$X_i$ (resp., $Y_i$), and thus $R$
intersects a set $G_i(R)$ of only a constant number of cells of $G_i$. Therefore, for
each $P_i\in \calP'$, only the surface patches of $\Ed_i(x,y)$ defined on the
cells of $G_i(R)$ are relevant for computing $q^*$. The supporting
planes of these surface patches can be determined immediately
after the above $O(\log m)$
recursive steps. By the same analysis as in Section
\ref{sec:firstprune}, all above can be done in $O(c\cdot m)$ time.

The above found $O(c)$ ``relevant'' supporting planes such that $q^*$
corresponds to a
lowest point in the upper envelope of them. Consequently, $q^*$ can be
found in $O(c)$ time by applying the linear-time algorithm for the 3D
LP problem \cite{ref:MegiddoLi83} on these $O(c)$ relevant supporting planes.
\qed
\end{proof}

This finishes our algorithm for computing $q^*$, which runs in
$O(mn)$ time.

\begin{theorem}\label{theo:2}
A rectilinear center $q^*$ of the uncertain points of $\calP$
in the plane can be computed in $O(mn)$ time.
\end{theorem}

\section{Concluding Remarks}\label{sec:conclude}

In this paper, we refine the prune-and-search technique \cite{ref:MegiddoLi83}
to solve in linear time the rectilinear one-center problem on uncertain points in the plane.
Note that the problem can also be considered as the one-center problem on uncertain points in the plane under the $L_1$ distance metric. Since the $L_{\infty}$ and $L_1$ metrics are closely related to each other (by rotating the coordinate axes by $45^{\circ}$), the same problem under the $L_{\infty}$ metric can be solved in linear time as well.

The Euclidean version of the problem seems more natural. Unfortunately, even if $\calP$ contains only one uncertain point $P_1$ and all locations of $P_1$ have the same probability, finding a center for $P_1$ is essentially the $1$-median problem in the plane, which is known as the Weber problem and no exact algorithm exists for it due to the computation challenge \cite{ref:BajajTh88}.

%\section{Conclusion}\label{Sec:5}
%
%In this paper, we refine the prune-and-search technique \cite{ref:MegiddoLi83}
%to solve the rectilinear one-center problem on uncertain points in $O(mn)$ time.
%Note that our algorithm is still applicable to the weighted version of this problem
%in which each $P_i$ is associated with a nonnegative weight $w_i$.
%It requires minimizing the maximum weighted expected rectilinear distance
%$R(x,y)=\max_{P_i\in\calP} w_i\Ed_i(x,y)$.
%Actually, we can reduce it to the unweighted version
%by changing each vale $f_{ij}$ to $w_i\cdot f_{ij}$
%for all $1\leq i\leq n$ and $1\leq j\leq m$.
%Then, we can apply our algorithm to solve the weighted version in $O(mn)$ time as well.

\bibliography{2DOneCenter}

\begin{thebibliography}{10}

\bibitem{ref:AgarwalIn09}
P.K. Agarwal, S.-W. Cheng, Y.~Tao, and K.~Yi.
\newblock Indexing uncertain data.
\newblock In {\em Proc. of the 28th Symposium on Principles of Database Systems
  (PODS)}, pages 137--146, 2009.

\bibitem{ref:AgarwalNe12}
P.K. Agarwal, A.~Efrat, S.~Sankararaman, and W.~Zhang.
\newblock Nearest-neighbor searching under uncertainty.
\newblock In {\em Proc. of the 31st Symposium on Principles of Database Systems
  (PODS)}, pages 225--236, 2012.

\bibitem{ref:AgarwalCo14}
P.K. Agarwal, S.~Har-Peled, S.~Suri, H.~Y{\i}ld{\i}z, and W.~Zhang.
\newblock Convex hulls under uncertainty.
\newblock In {\em Proc. of the 22nd Annual European Symposium on Algorithms
  (ESA)}, pages 37--48, 2014.

\bibitem{ref:BajajTh88}
C.~Bajaj.
\newblock The algebraic degree of geometric optimization problems.
\newblock {\em Discrete and Computational Geometry}, 3:177--191, 1988.

\bibitem{ref:BrassTh11}
P.~Brass, C.~Knauer, H.-S. Na, C.-S. Shin, and A.~Vigneron.
\newblock The aligned {$k$}-center problem.
\newblock {\em International Journal of Computational Geometry and
  Applications}, 21:157--178, 2011.

\bibitem{ref:ChanMo99}
T.M. Chan.
\newblock More planar two-center algorithms.
\newblock {\em Computational Geometry: Theory and Applications},
  13(3):189--198, 1999.

\bibitem{ref:CormodeAp08}
G.~Cormode and A.~McGregor.
\newblock Approximation algorithms for clustering uncertain data.
\newblock In {\em Proc. of the 27th Symposium on Principles of Database Systems
  (PODS)}, pages 191--200, 2008.

\bibitem{ref:BergKi13}
M.~de~Berg, M.~Roeloffzen, and B.~Speckmann.
\newblock Kinetic 2-centers in the black-box model.
\newblock In {\em Proc. of the 29th Annual Symposium on Computational Geometry
  (SoCG)}, pages 145--154, 2013.

\bibitem{ref:DongDa07}
X.~Dong, A.Y. Halevy, and C.~Yu.
\newblock Data integration with uncertainty.
\newblock In {\em Proceedings of the 33rd International Conference on Very
  Large Data Bases}, pages 687--698, 2007.

\bibitem{ref:FoulA106}
A.~Foul.
\newblock A $1$-center problem on the plane with uniformly distributed demand
  points.
\newblock {\em Operations Research Letters}, 34(3):264--268, 2006.

\bibitem{ref:HoffmannAs05}
M.~Hoffmann.
\newblock A simple linear algorithm for computing rectilinear 3-centers.
\newblock {\em Computational Geometry}, 31(3):150--165, 2005.

\bibitem{ref:KamousiCl11}
P.~Kamousi, T.M. Chan, and S.~Suri.
\newblock Closest pair and the post office problem for stochastic points.
\newblock In {\em Proc. of the 12nd International Workshop on Algorithms and
  Data Structures (WADS)}, pages 548--559, 2011.

\bibitem{ref:KamousiSt11}
P.~Kamousi, T.M. Chan, and S.~Suri.
\newblock Stochastic minimum spanning trees in {Euclidean} spaces.
\newblock In {\em Proc. of the 27th Annual Symposium on Computational Geometry
  (SoCG)}, pages 65--74, 2011.

\bibitem{ref:KarmakarSo13}
A.~Karmakar, S.~Das, S.C. Nandy, and B.K. Bhattacharya.
\newblock Some variations on constrained minimum enclosing circle problem.
\newblock {\em Journal of Combinatorial Optimization}, 25(2):176--190, 2013.

\bibitem{ref:MegiddoLi83}
N.~Megiddo.
\newblock Linear-time algorithms for linear programming in {$R^3$} and related
  problems.
\newblock {\em SIAM Journal on Computing}, 12(4):759--776, 1983.

\bibitem{ref:MegiddoOn84}
N.~Megiddo and K.J. Supowit.
\newblock On the complexity of some common geometric location problems.
\newblock {\em SIAM Journal on Comuting}, 13:182--196, 1984.

\bibitem{ref:SuriOn14}
S.~Suri and K.~Verbeek.
\newblock On the most likely voronoi diagram and nearest neighbor searching.
\newblock In {\em Proc. of the 25th International Symposium on Algorithms and
  Computation (ISAAC)}, pages 338--350, 2014.

\bibitem{ref:SuriOn13}
S.~Suri, K.~Verbeek, and H.~Y{\i}ld{\i}z.
\newblock On the most likely convex hull of uncertain points.
\newblock In {\em Proc. of the 21st European Symposium on Algorithms (ESA)},
  pages 791--802, 2013.

\bibitem{ref:WangLi14}
H.~Wang and J.~Zhang.
\newblock Line-constrained $k$-median, $k$-means, and $k$-center problems in
  the plane.
\newblock In {\em Proc. of the 25th International Symposium on Algorithms and
  Computation (ISAAC)}, pages 104--115, 2014.

\bibitem{ref:WangCo15}
H.~Wang and J.~Zhang.
\newblock Computing the center of uncertain points on tree networks.
\newblock In {\em Proc. of the 14th Algorithms and Data Structures Symposium
  (WADS)}, pages 606--618, 2015.

\bibitem{ref:WangOn15}
H.~Wang and J.~Zhang.
\newblock One-dimensional {$k$}-center on uncertain data.
\newblock {\em Theoretical Computer Science}, online first, 2015.

\end{thebibliography}
\bibliographystyle{plain}

\end{document}